\documentclass[
]{CEURART/ceurart}

\usepackage{bbding}
\usepackage{comment}
\usepackage{stmaryrd}
\usepackage{forest}
\usepackage{longtable}
\newcolumntype{H}{@{}>{\setbox0=\hbox\bgroup}c<{\egroup}@{}}
\usepackage{todonotes}

\usepackage{listings}
\usepackage{tikz}
\usepackage{circuitikz}
\usepackage{float}
\usepackage{bussproofs}

\usepackage{subcaption}
\usepackage{adjustbox}

\newtheorem{theorem}{Theorem}[section]

\newtheorem{corollary}[theorem]{Corollary}
\newtheorem{defin}[theorem]{Definition}
\newtheorem{example}{Example}[section]

\newproof{proof}{Proof}

\lstdefinelanguage{scala}{
  alsoletter={@,=,>},
  morekeywords={abstract, case, class, def, do, Input, Output, then,
        else, extends, false, free, if, implicit, match,
        object, true, val, var, while, sealed, or,
        for, dependent, null, type, with, try, catch, finally,
        import, final, return, new, override, this, trait,
        private, public, protected, package, throw, enum, require, assert},
  sensitive=true,
  morecomment=[l]{//},
  morecomment=[s]{/*}{*/},
  morestring=[b]",
}
  {\gdef\scalefactor{#1}\begin{center}\proofSkipAmount \leavevmode}%
  {\resizebox{\scalefactor}{!}{\DisplayProof}\proofSkipAmount \end{center} }

\renewcommand{\O}{\mathcal O}
\renewcommand{\P}{\mathcal P}

  \newcommand{\letspace}{{\kern-0.1em}}

  \newcommand{\logicstyle}[1]{#1}

  \newcommand{\OL}{\ensuremath{\logicstyle{O{\letspace}L}}\xspace}

  \newcommand{\NFOL}{\NF_\OL\xspace}

\DeclareMathOperator{\NF}{NF}

\DeclareMathOperator{\nnf}{NNF}
\DeclareMathOperator{\getInverse}{getInverse}

\renewcommand{\slash}[1]{{{\kern-0.05em}/{\kern-0.08em}{#1}}}

\newcommand{\simw}[1]{\sim_{{\kern-0.1em}#1}}

\newcommand{\term}[2]{\mathcal{T}_{#2}%
    \ifthenelse{\isempty{#1}}%
    {}%
    {(#1)}%
}
\newcommand{\optarg}[2]{%
  \ifthenelse{\isempty{#1}}%
    {}
    {(((#1)))}
}

\usepackage{framed}
\usepackage{mathpartir}
\usepackage{algorithm}
\usepackage[noend]{algpseudocode}

\makeatletter
\let\saved@ref\ref
\renewcommand{\ref}{\begingroup\catcode`\_=8\relax\ref@inner}
\newcommand{\ref@inner}[1]{\saved@ref{#1}\endgroup}
\let\saved@autoref\autoref
\renewcommand{\autoref}{\begingroup\catcode`\_=8\relax\autoref@inner}
\newcommand{\autoref@inner}[1]{\saved@autoref{#1}\endgroup}
\let\saved@label\label
\renewcommand{\label}{\begingroup\catcode`\_=8\relax\label@inner}
\newcommand{\label@inner}[1]{\saved@label{#1}\endgroup}
\makeatother

\newcommand{\smartparagraph}[1]{\noindent\textbf{#1}}

\title{Orthologic for SAT Solving}
\author[1]{Vladislas de Haldat}
\author[1]{Simon Guilloud}[orcid=0000-0001-8179-7549,email=simon.guilloud@epfl.ch]
\author[1]{Viktor Kun\v{c}ak}[orcid=0000-0001-7044-9522]
\address[1]{EPFL}

\copyrightyear{2026}
\copyrightclause{Copyright for this paper by its authors.
  Use permitted under Creative Commons License Attribution 4.0
  International (CC BY 4.0).}
\conference{}

\begin{document}
\addtolength{\textfloatsep}{-0.2in}
\sloppy

\maketitle
\begin{abstract}
We present a new algorithm for deciding formula entailment in orthologic (a sound approximation of classical logic) that avoids the costly preprocessing phase of prior implementations while retaining the same $\mathcal{O}(n^2(1+|A|))$ worst-case complexity. We then introduce a family of synthetic SAT benchmarks based on the observation that, for any formula $\phi$, the equivalence $\phi \leftrightarrow \mathrm{NF}_{\mathrm{OL}}(\phi)$ is a tautology whose Tseitin encoding yields unsatisfiable instances that are hard for state-of-the-art SAT solvers yet have short orthologic proofs. Applied to EPFL arithmetic circuits, our algorithm solves these instances efficiently while Kissat times out on a significant fraction. Finally, we show that using orthologic normalization as a preprocessing step can improve SAT solving time on some hard problems.
\end{abstract}

\section{Introduction}
\label{sec:introduction}
Satisfiability checking has many applications including verification, testing, and optimization. It is also the prototypical NP-complete problem, which gives little hope that it admits a (complete) polynomial-time algorithm.
SAT solvers have made significant strides in efficiency and capability through optimizations, simplification of formulas, and heuristics. However, certain classes of seemingly easy problems remain challenging even for the state-of-the-art SAT solvers.

Recently, researchers have proposed \emph{orthologic} as a complementary tool for reasoning in propositional logic  \cite{DBLP:conf/cav/GuilloudBMK23,guilloudOrthologicAxiomsErrata2024}. Orthologic is a weakening of classical logic that relaxes the distributive law of conjunction and disjunction, making it a sound but incomplete approximation of classical logic. The main interest in orthologic stems from its good algorithmic properties: it is possible to decide if a formula is true or false in quadratic time, in contrast to the (co)NP-completeness of the corresponding problems for classical logic. 
Concretely, orthologic admits a quadratic-time normalization procedure and an $\mathcal{O}(n^2|A|)$-time algorithm for the entailment problem (deciding validity under a set of ground assumptions)\cite{guilloudOrthologicAxiomsErrata2024}. 

These algorithmic properties suggest that orthologic can be used as an efficient and predictable building block for reasoning systems. This approach has already been implemented in the domains of program verifiers \cite{DBLP:conf/cav/GuilloudBMK23} and of proof assistants \cite{guilloudLISAModernProof2023}, where predictability is especially valuable. In this paper we argue that there are propositional satisfiability problems that fall in the domain of orthologic but which are non-trivial for existing state-of-the-art SAT solvers, suggesting that orthologic-based reasoning can become a useful technique in the toolbox of SAT solvers. 
This view is further supported by the close correspondence between orthologic and established circuit simplification techniques. For example, 15 of the 16 local rewriting rules for And-Inverter Graphs (AIGs) proposed by Brummayer and Biere~\cite{brummayerLocalTwoLevelAndInverter2006} are subsumed by orthologic reasoning. Moreover, orthologic provides a notion of formula equivalence—OL-equivalence—for which detecting equivalent sub-circuits can be decided in quadratic time, offering an efficient alternative to the congruence-closure-based approaches used in state-of-the-art preprocessors~\cite{biereClausalCongruenceClosure2024}.

Practical challenges in the usage of orthologic have existed in the past. While quadratic time-algorithms for the word problem (which is a special case of deciding if two expressions are equivalent without non-logical axioms) were known for longer \cite{brunsFreeOrtholattices1976}, known algorithms for the entailment problem had polynomial coefficients that were too high for practical use (for example, $\mathcal O(n^7)$ \cite{DBLP:journals/sLogica/EglyT03}). An $\mathcal{O}(n^2(1+|A|))$ algorithm for entailment in orthologic (where $|A|$ is the number of assumptions) was presented in \cite{guilloudOrthologicAxiomsErrata2024} based on a proof search procedure in a variant of sequent calculus, but the initially-published backward proof-search algorithm contained a subtle error related to memoization that made it incomplete. A different algorithm, reducing orthologic proof search to solving a system of Horn clauses, was described as a replacement \cite{guilloudOrthologicAxiomsErrata2024}. The Horn clause reduction has significant issues for practical applicability to large problems. Indeed, in the reduction, one literal is created for every pair of subformulas of the input, representing the implication $\phi \vdash \psi$.  The Horn clauses, of which there are at most $\mathcal{O}(n^2|A|)$, correspond to the inference rules of the proof system. Since systems of Horn clauses can be solved in linear time, the result follows. The major practical issue is that, before even trying a single deduction of new inequalities, the algorithm generates these $\mathcal{O}(n^2|A|)$ clauses. Consequently, the algorithm will always have worst-case behaviour, even if the problem turns out to have a simple proof with a linear or constant number of deductions (for example if the problem contains many irrelevant assumptions). 

In the present work, we address these shortcomings with a new algorithm for deciding the entailment problem, and explore further connections between orthologic and SAT solving. The new algorithm is complete and achieves the best-known theoretical worst-case complexity, while avoiding the pre-computation of all the possible inferences between orthologic inequalities (corresponding to Horn clauses in the original algorithm mentioned above).

We then present a family of valid propositional formulas that orthologic can certify efficiently but which remain challenging for SAT solvers. For any formula $\phi$, the equivalence $\phi \leftrightarrow \NFOL(\phi)$ is a tautology (where $\NFOL(\phi)$ denotes the orthologic normal form of $\phi$), yet its Tseitin encoding yields an unsatisfiable SAT instance on which state-of-the-art solvers such as Kissat time out for large arithmetic circuits. We apply this construction to a benchmark of circuits computing arithmetic functions and show that our new algorithm for orthologic entailment solves them efficiently whereas Kissat fails on a significant fraction. Since orthologic proofs are at most quadratic in size, this means that short proofs exist but are beyond the reach of SAT solver heuristics. This construction thus suggests a new way to generate synthetic benchmarks that are unsatisfiable yet hard for SAT solvers.

Finally, we study the applicability of orthologic normalization as a preprocessing tool for deciding satisfiability of Boolean circuits. Using the orthologic normalization function, we can simplify a formula to the smallest \OL-equivalent one. This idea was previously suggested and implemented in the context of program verification with the Stainless verifier \cite{DBLP:conf/cav/GuilloudBMK23,HamzaETAL19SystemFR}, but had not been evaluated in the context of SAT solving.
We present experimental results comparing the time taken by SAT solvers on a set of standard public benchmarks \cite{velev,BiereAigerBenchmarks,tjunttil} with and without orthologic normalization. Our results show that in the vast majority of benchmarks, orthologic normalization makes the resulting problem faster to solve.
Some problems that time out without the normalization also become solvable after the normalization. These results confirm that orthologic-based reasoning can positively contribute to practical SAT solving.

The contributions of this paper are as follows:
\begin{itemize}
    \item We present (\autoref{sec:proofsearch}) a new algorithm for the entailment problem in orthologic that maintains the theoretical best-known $\mathcal O(n^2(1+|A|))$ time complexity but in contrast with previously-described implementations avoids a costly preprocessing phase, allowing for further optimizations in the deduction of new inequalities. We also show that our improved algorithm is orders of magnitude faster than the previously described Horn-clause-based implementation.
    \item We then show (\autoref{sec:olBenchmarks}) a new way to generate synthetic benchmarks from Boolean circuits that are difficult even for state-of-the-art SAT solvers, but which have short (at most quadratic) proofs in orthologic.
    \item We evaluate (\autoref{sec:preprocessing}) the impact of orthologic normalization on SAT solving. We show using large published benchmarks that running a SAT solver on the OL-normalized version of a circuit can improve runtime by up to 28\%.
\end{itemize}

\smartparagraph{Further Related Work.}
Simplification of Boolean circuits as a preprocessing technique for SAT solving has been studied on many occasions. We highlight a selection. 
\cite{francesdemasBinaryImplicationHypergraphs2024} describes an equivalence-preserving, size-reducing simplification rule based on a hypergraph representation. \cite{brummayerLocalTwoLevelAndInverter2006} present a set of rules for optimizing circuits represented as And-Inverter Graphs (AIGs). The rules are local in the sense that they are all rewrite rules of depth at most two. This is in contrast with lattice and ortholattice-based simplification, which cannot be represented as a complete rewriting system.
\cite{guilloudEquivalenceCheckingOrthocomplemented2022} use a term rewriting system based on orthocomplemented bisemilattices, an algebraic variety related to ortholattices but strictly weaker, to devise a linear-time simplification algorithm for Boolean circuits. 
\cite{bjesseDAGawareCircuitCompression2004} present and evaluate a fast (linear-time) compression algorithm using auto-compressing representation of circuits that merges features of other preexisting optimizations.
\cite{biereClausalCongruenceClosure2024} uses congruence closure to detect isomorphic sub-circuits, and eliminates the corresponding nodes. This reduces the total size of the problem after clausification and improves SAT-solvers performances.

The word problem for ortholattices was first solved in \cite{brunsFreeOrtholattices1976}. The proof system of orthologic on which we base our improved algorithm was first presented in \cite{schultemontingCutEliminationWord1981}. The cut elimination theorem was formalized in Rocq in \cite{guilloudVerifiedOptimizedImplementation2025}. The normalization algorithm we study was presented in \cite{DBLP:conf/cav/GuilloudBMK23}. Other treatments of orthologic proof systems and proof search include \cite{kawanoLabeledSequentCalculus2018, laurentFocusingOrthologic2016}, but with much worse complexity results.

\section{Preliminaries}
\label{sec:preliminaries}

We next introduce the key concepts and the notation.

\begin{defin}
    \textbf{Ortholattices} are algebras with signature $(\land^2, \lor^2, \neg^1, \bot^0, \top^0)$, where superscript indicates arity, that satisfy the laws given in \autoref{tab:axVariety}. In particular, the class of all ortholattices is a variety.

    Ortholattices are in particular \textit{lattices}, and hence admit a natural order given by
    $
    a \leq b \iff a \land b = a
    $,
    which is also provably equivalent to $a \lor b = b$. 
\end{defin}

\begin{table}[hbt]
\centering
    \begin{tabular}{ c @{\hskip 2em} | @{\hskip 2em} c }
        $x\vee y=y\vee x$&$x\wedge y=y\wedge x$\\
        $x\vee(y\vee z)=(x\vee y)\vee z$&$x\wedge(y\wedge z)=(x\wedge y)\wedge z$\\
        $x\vee x=x$&$x\wedge x=x$\\
        $x\vee\top=\top$&$x\wedge\bot=\bot$\\
        $x\vee\bot=x$&$x\wedge\top=x$\\
        $\neg\neg x=x$& \\
        $x\vee\neg x=\top$&$x\wedge\neg x=\bot$\\
        $\neg(x\vee y)=\neg x\wedge\neg y$&$\neg(x\wedge y)=\neg x\vee\neg y$\\
        $x\vee(x\wedge y)=x$&$x\wedge(x\vee y)=x$
    \end{tabular}
    \caption{Axiomatization of ortholattices, algebras with signature $(\land^2, \lor^2, \neg^1, \bot^0, \top^0)$}
    \label{tab:axVariety}
\end{table}

Ortholattices are a weakening of Boolean algebras. Boolean algebras are precisely distributive ortholattices, meaning satisfying the two distributivity laws:
\begin{equation*}
\begin{array}{c}
a \lor (b \land c) = (a \lor b) \land (a \lor c) \\
a \land (b \lor c) = (a \land b) \lor (a \land c)
\end{array}
\end{equation*}

\subsection{Orthologic Proof System}

Classical logic is the logic whose Tarski-Lindenbaum algebra is Boolean algebra. Similarly, orthologic is the logic whose algebra is that of ortholattices. 
\begin{defin}
    Let $X$ be some countably infinite set of variable symbols. We denote by $\term{X}{\OL}$ the set of terms built from $X$ and the signature of ortholattices. Note that this is the same set of expressions as propositional formulas in classical logic, so we sometimes call elements of $\term{X}{\OL}$ \textit{formulas}.

    We note $\phi \simw{\OL} \psi$ if both $\phi \leq \psi$ and $\psi \leq \phi$ hold in all ortholattices.
\end{defin}
Orthologic admits a simple and natural proof system, similar to sequent calculus for classical and intuitionistic logic. 
It is folklore that by restricting the sequent calculus LK for classical logic so that sequents only have at any point in a proof at most one formula on the right side, we obtain a proof system (named LJ) for intuitionistic logic. A similar restriction characterizes orthologic: Sequents can never contain more than two formulas in total (but can contain two on the right).
\begin{defin}
    If $\phi \in \term{X}{\OL}$ is a formula, then we call $\phi^L$ and $\phi^R$ annotated formulas. An orthologic \textit{sequent} is a pair of annotated formulas $(\Gamma, \Delta)$. The set of rules of the orthologic proof system is given in \autoref{fig:olProofSystem}.
\end{defin}

Because of the restriction that sequents contain at most two formulas, it is convenient to use annotated formulas rather than a pair of sets. $\phi^L$ denotes that $\phi$ is left of the sequent, and similarly $\psi^R$ denotes that $\psi$ is on the right. Hence, $\phi^L, \psi^R$ denotes the sequent $\phi \vdash \psi$ in more conventional notation.
In practice, for implementation, it is convenient to consider orthologic sequents as containing exactly two formulas rather than at most two: A sequent with a single annotated formula $\Gamma$ can be written as $\Gamma, \Gamma$ while the empty sequent can be (for example) $\bot^R, \bot^R$.

\begin{figure}[htbp]
    \begin{framed} 
        \[
            \inferrule*[Right=(Cut)]{\Gamma,\phi^R\\\phi^L,\Delta}{\Gamma,\Delta}
        \]
        \begin{align*}
        \inferrule*[Right=(Hyp)]{ }{\phi^L,\phi^R}
        &\hspace{8em}
        \inferrule*[Right=(Replace)]{\Gamma, \Gamma}{\Gamma,\Delta}\\
        \inferrule*[Right=($\wedge$-$L$)]{\phi^L,\Delta}{(\phi\wedge\psi)^L,\Delta}
        &\hspace{8em}
        \inferrule*[Right=($\wedge$-$R$)] 
        {\Gamma,\phi^R\\\Gamma,\psi^R}
        {\Gamma,(\phi\wedge\psi)^R}\\
        \inferrule*[Right=($\vee$-$L$)]
        {\phi^L,\Delta\\\psi^L,\Delta}
        {(\phi\vee\psi)^L,\Delta}
        &\hspace{8em}
        \inferrule*[Right=($\vee$-$R$)]
        {\Gamma,\phi^R}{\Gamma,(\phi\vee\psi)^R}\\
        \inferrule*[Right=($\neg$-$L$)]
        {\Gamma,\phi^R}{\Gamma,(\neg\phi)^L}
        &\hspace{8em}
        \inferrule*[Right=($\neg$-$R$)]
        {\phi^L,\Delta}{(\neg\phi)^R,\Delta}
    \end{align*}
        \caption{Proof system for propositional orthologic with axioms}
        \label{fig:olProofSystem}
    \end{framed}
\end{figure}
\begin{theorem}[Soundness and completeness \cite{guilloudOrthologicAxiomsErrata2024}]
\label{thm:soundComplete}
    A sequent $\phi^L, \psi^R$ is provable in orthologic if and only if the inequality $\phi \leq \psi$ holds in every ortholattice.
\end{theorem}

Even though it is strictly weaker than classical logic, orthologic allows proving a number of desirable properties, including all laws of bounded lattices, such as reordering and duplicating conjuncts and disjuncts, constant folding of $\bot$ and $\top$, and all laws necessary to compute the negation normal form (nnf) of a formula.

Note that by the laws of ortholattices, $\bot$ and $\top$ can always be written as respectively $x \land \neg x$ and $x \lor \neg x$, for any variable $x$. Hence, for simplicity, we may sometimes omit cases relative to $\bot$ and $\top$ in proofs.

Orthologic admits strong algorithmic properties: In particular, the proof system of \autoref{fig:olProofSystem} admits a polynomial-time decision procedure. Formally, we are interested in the following three kinds of problems:
\begin{defin}
    The \emph{word problem} for ortholattices consists in deciding whether a given inequality $\phi \leq \psi$ holds in all ortholattices. This is equivalent to $\phi^L, \psi^R$ being provable. 
    The \emph{entailment problem} (also known as generalized word problem) consists in deciding, given an inequality $\phi \leq \psi$ and a set of inequalities $A$, if $\phi \leq \psi$ holds in \textit{all ortholattices which satisfy all the inequalities of $A$}. Elements of $A$ are called non-logical axioms. The entailment problem is equivalent to deciding if $\phi^L, \psi^R$ is provable in orthologic with the additional rules for every $(\Gamma, \Delta) \in A$:
    $$
    \inferrule*[Right=(Ax)]{\phantom{x}}{\Gamma,\Delta}
    $$

    A \textit{normalization function} is a function $f: \term{X}{\OL} \to \term{X}{\OL}$ such that $\phi \simw{\OL} \psi$ if and only if $f(\phi) = f(\psi)$. 
\end{defin}

\begin{theorem}
The word problem for ortholattices is solvable in time $\mathcal{O}(n^2)$ \cite{brunsFreeOrtholattices1976}. 
The entailment problem is solvable in time $\mathcal{O}(n^2|A|)$ \cite{guilloudOrthologicAxiomsErrata2024}. 
There exists a normalization function for ortholattices computable in time $\mathcal{O}(n^2)$ \cite{DBLP:conf/cav/GuilloudBMK23}. Given a formula $\phi$, let $\NFOL(\phi)$ denote its normal form. Additionally, this function maps every formula to the smallest (by number of AND and OR connectives) formula in its equivalence class.
\end{theorem}

The word problem is a special case of the entailment problem (take $A = \emptyset$) and can also be decided using a normalization function (check whether $f(\phi) = f(\psi)$). For practical applications, supporting non-logical axioms allows reasoning with additional knowledge, and normalization allows progressing on formulas even when they don't entirely fall in the domain of orthologic. 

The existence of an $\mathcal{O}(n^2|A|)$ algorithm for the entailment problem in ortholattices follows from the following key property.
\begin{theorem}[Partial Cut Elimination Theorem \cite{guilloudOrthologicAxiomsErrata2024}]
If a sequent $(\Gamma, \Delta)$ has a proof in orthologic with axioms in $A$, then it has a proof where in all the instances of the \texttt{Cut} rule, the cut formula $\phi$ is in one of the axioms. In particular, if the set of axioms is empty, then a sequent has a proof if and only if it has a proof without the \texttt{Cut} rule.
\end{theorem}
Since the \texttt{Cut} rule is the only one that can remove a formula from a proof, partial cut elimination implies the following important corollary.
\begin{corollary}[Subformula Property]
If a sequent $(\Gamma, \Delta)$ has a proof from axioms in $A$, then it has a proof which only uses subformulas of $\Gamma, \Delta$ and of the axioms in $A$.
\end{corollary}

\smartparagraph{Reduction to HornSAT.}
The subformula property yields the following reduction to satisfiability of propositional Horn clauses \cite{DOWLING1984267}.  Observe that there exists at most $\mathcal{O}(n^2)$ different sequents that can be built from the subformulas of $\Gamma, \Delta$ and of the axioms in $A$. Create a propositional variable for each such sequent. Create a Horn clause for every instance of a deduction rule in \autoref{fig:olProofSystem}. Every sequent can be deduced in a constant number of different ways for every rule but the \texttt{Cut} rule, and it can be deduced using the \texttt{Cut} rule in as many ways as there are possible cut formulas, that is $|A|$. Hence, there are at most $\mathcal{O}(n^2|A|)$ total Horn clauses. Systems of Horn clauses can be solved in linear time \cite{DOWLING1984267}, so it follows that entailment in orthologic can be solved in time $\mathcal{O}(n^2|A|)$.

\section{Better Proof Search Algorithm For OL with Axioms}
\label{sec:proofsearch}
The Horn-clause-based algorithm for the entailment problem has a significant practical issue: It always precomputes $\mathcal{O}(n^2|A|)$ Horn clauses before making a single deduction. Hence, it essentially always hits the worst case complexity, even if the problem happened to have a one-step proof!
This makes the algorithm rather impractical and slow in practice. It also prevents us from applying any heuristics on how to deduce new sequents, as the complexity is anyway dominated by the preprocessing phase.
In this section, we present a new algorithm that avoids this issue. Additionally, our new algorithm has only $\mathcal{O}(n^2)$ worst-case space complexity instead of $\mathcal{O}(n^2|A|)$.

\subsection{Improved algorithm for orthologic entailment}
\label{sec:improvedOLAlgo}
First, it is convenient to compute the negation normal form ($\nnf$) of formulas, restricting their syntax. For propositional formulas, negation normal form is easy to compute in linear time. Note that in order for the algorithm to still run in linear time in the presence of structure sharing, we need to memoize the result of $\nnf$ for each subformula.
Moreover, we will need later to compute the negation of a formula in negation normal form. Formally, we define
$$
\begin{array}{l}
\getInverse: \term{X}{\OL} \to \term{X}{\OL}\\
\getInverse(\phi) := \nnf(\neg \phi)
\end{array}
$$

We use $\phi'$ as a shorthand for $\getInverse(\phi)$.

For any formula $\phi \in \term{X}{\OL}$, $\nnf(\phi) \simw{\OL} \phi$, so to decide if a sequent $\phi^{\square}, \psi^{\triangle}$ is provable in $\OL$, it is sufficient to decide if the sequent $\nnf(\phi)^{\square}, \nnf(\psi)^{\triangle}$ is provable. Hence, we can restrict our attention to formulas in negation normal form.
Additionally, by soundness and completeness of the $\OL$ proof system (\autoref{thm:soundComplete}), a sequent $(\phi^L, \Delta)$ is provable if and only if $((\neg\phi)^R, \Delta)$ is provable. Hence, any sequent can be transformed, in linear time, into an equiprovable sequent with two right formulas so that negation only appears right on top of variables.

Hence, by replacing the \textsc{Hyp} rule by the equivalent

\begin{center}
    \AxiomC{}
    \RightLabel{\textsc{Hyp}}
    \UnaryInfC{$x^R, \neg x^R$}
    \DisplayProof
\end{center}
and the \textsc{Cut} rule by
\begin{center}
    \AxiomC{$\phi^R, \gamma^R$}
    \AxiomC{$\getInverse(\gamma)^R, \psi^R$}
    \RightLabel{\textsc{Cut}}
    \BinaryInfC{$\phi^R, \psi^R$}
    \DisplayProof
\end{center}
we can ensure all sequents in a proof only ever have right formulas; in particular we can ignore the rules ($\vee$-$L$), ($\wedge$-$L$), ($\neg$-$L$) and ($\neg$-$R$). This means that our search space is not all pairs of left and right subformulas, but pairs of the subformulas of the $\nnf$ of the input and the $\nnf$ of their negations. From now on, we consider without loss of generality that sequents have only right formulas that are in $\nnf$.
Hence, the rules that we need to consider are \textsc{Ax}, \textsc{Hyp}, \textsc{Cut}, \textsc{Replace}, $\wedge$-$R$, $\vee$-$R$.
\newcommand{\SF}{\mathit{S{\kern-0.1em}F}}
\newcommand{\AF}{\mathit{A{\kern-0.1em}F}}
We denote by $\SF$ the set of subformulas of our problem and by $\AF$ the set of formulas that compose the axioms. By $P$, we denote the set of proven sequents, initially containing only the axioms and hypotheses.

The goal is to design an algorithm that can compute up to $\mathcal{O}(n^2)$ provable sequents in time at most $\mathcal{O}(n^2(1+|A|))$ while limiting preprocessing as much as possible.
To get an intuition for the problem, observe that if there are a constant number of axioms, then our asymptotic time allowance and number of deductions are equal; hence, we may generally only perform a constant number of operations for each deduced sequent, except when the \textsc{Cut} rule is involved.

Suppose we just proved the sequent $(\phi, \psi)$. We want to discover every sequent that follows from it and previously proven sequents by some rule. We explain informally how the adequate data structure for each rule is maintained before giving the full algorithm.

For \textsc{Ax} and \textsc{Hyp}, we simply start the algorithm by adding to the proven list every axiom and every sequent of the form $(x, \neg x)$, for every variable $x$ in the input.

For the \textsc{Cut} rule, that means finding all previously proven sequents that contain either $\getInverse(\phi)$ or $\getInverse(\psi)$. To do so, we maintain a map
$P_{Cut}: (\AF \cup \AF') \to \mathcal P(\SF)$
from axiom formulas to sets of formulas. $P_{Cut}$ satisfies the following invariant:
$$
P_{Cut}(a) = \lbrace b \mid (\getInverse(a), b) \in P \rbrace
$$
When proving $(\phi, \gamma)$, if $\gamma \in \AF$, then it follows that for every formula $\psi \in P_{Cut}(\getInverse(\gamma))$, $(\phi, \psi)$ is provable. Note that it may be that $(\phi, \psi)$ was already deduced from another cut formula, but there are at most $|A|$ such cut formulas.

Then, the ($\wedge$-$R$) rule has two premises.
First, we initialize a map
$$
\SF_{\land}: \SF \to \mathcal P(\SF)
$$
$$
\SF_{\land}(a) = \lbrace b \mid a \land b \in \SF \rbrace
$$

We also maintain a different map
$$
P_{\land}: \SF \to \SF \to \mathcal P(\SF)
$$
with the invariant
$
P_{\land}(b)(k) = \lbrace a\land k \mid (a, b) \in P \text{ and } a \land k \in \SF  \rbrace
$.

When proving a sequent $(a, b)$, for every $k \in \SF_{\land}(a)$, we update
$$
P_{\land}(b)(k)\, \mathop{\texttt{+=}}\, a \land k,
$$
and symmetrically for $b$. Then, when deducing a sequent $(\phi, \psi)$, we can automatically deduce all sequents of the form $(\phi, \gamma)$ for $\gamma \in P_{\land}(\phi)(\psi)$.

The ($\vee$-$R$) rule is even simpler, as it has only one premise. We initialize a map
$$
\SF_{\lor}: \SF \to \mathcal P(\SF)
$$
$$
\SF_{\lor}(a) = \lbrace a \lor b \mid a \lor b \in \SF \rbrace
$$
When deducing a sequent $(\phi, \psi)$, then for all $\gamma \in \SF_{\lor}(\psi)$ we should also deduce the sequent $(\phi, \gamma)$ (and similarly for $\phi$).

For the \textsc{Replace} rule, which also has only one premise, when deducing a sequent $(\phi, \phi)$, we can deduce the sequent $(\psi, \phi)$ for every $\psi \in \SF$.

\autoref{alg:olEntailment} summarizes the algorithm. In it, every data structure we use is the mutable version and \lstinline|X| is the set of variables. Moreover, pairs of formulas are unordered pairs. The following theorem states its correctness.

\begin{algorithm}[p]
\caption{Optimized Orthologic Entailment Algorithm}\label{alg:olEntailment}
\begin{lstlisting}[language=scala, mathescape=true]
type Sequent = (Formula, Formula)
def prove(s: Sequent, axioms: Set[Sequent]): Boolean =
  val subformulas: Set[Formula] = subformulasOf(s) ++ axioms.flatMap(subformulasOf)
  subformulas.toArray.foreach(f => subformulas ++= subformulasOf(getInverse(f)))
  val axiomFormulas: Set[Formula] = 
          axioms.flatMap{ case (a, b) => Set(a, b, getInverse(a), getInverse(b)) }
  val proven: Set[Sequent] = Stack.from(axioms ++ X.map(x => (x, getInverse(x))))
  val worklist: Stack[Sequent] = Stack.from(axioms ++ X.map(x => (x, getInverse(x))))
  val P_Cut: Map[Formula, Set[Formula]] = Map.empty
  val P_and: Map[Formula, Map[Formula, Set[Formula]]] = Map.empty

  val SF$_\lor$: Map[Formula, Set[Formula]] = Map.empty
  val SF$_\land$: Map[Formula, Set[Formula]] = Map.empty
  subformulas.foreach :
    case $\varphi$: $\varphi_1$ /\ $\varphi_2$ =>
      SF$_\land$($\varphi_1$) += $\varphi_2$
      SF$_\land$($\varphi_2$) += $\varphi_1$
    case $\varphi$: $\varphi_1$ \/ $\varphi_2$ =>
      SF$_\lor$($\varphi_1$) += $\varphi$
      SF$_\lor$($\varphi_2$) += $\varphi$

  while worklist.nonEmpty do
    val (a, b) = worklist.pop()
    if (a, b) == s then return true
    if axiomFormulas.contains(a) then P_Cut(getInverse(a)) += b
    if axiomFormulas.contains(b) then P_Cut(getInverse(b)) += a

    SF$_\land$(a).foreach{ $\psi$ => P_and(b)($\psi$) += a /\ $\psi$}
    SF$_\land$(b).foreach{ $\psi$ => P_and(a)($\psi$) += b /\ $\psi$}

    (SF$_\lor$(a) ++ P_and(b)(a) ++ P_Cut(a)).foreach: $\varphi$ =>
      if !proven.contains(($\varphi$, b)) then
        proven.add(($\varphi$, b))
        worklist.push(($\varphi$, b))
    (SF$_\lor$(b) ++ P_and(a)(b) ++ P_Cut(b)).foreach: $\varphi$ =>
      if !proven.contains((a, $\varphi$)) then
        proven.add((a, $\varphi$))
        worklist.push((a, $\varphi$))

    if a == b then
      subformulas.foreach: $\varphi$ =>
        if !proven.contains(($\varphi$, a)) then
          proven.add(($\varphi$, a))
          worklist.push(($\varphi$, a))

  return false
\end{lstlisting}
\end{algorithm}

\begin{theorem}[Correctness of $\OL$ entailment algorithm]
    The \lstinline|prove| function of \autoref{alg:olEntailment} outputs \lstinline|true| if and only if the sequent $s$ has a proof from the axioms in $A$.
\end{theorem}
\begin{proof}
    ($\Rightarrow$) Suppose the algorithm outputs \lstinline|true| on the sequent $s$ and the axioms $A$. Let us prove that the sequent $s$ has a proof from axioms $A$. Before making any deductive step, the algorithm initializes the data structures. The sets $proven$, $P_{Cut}$ and $P_\land$ are initialized as empty and the sets $SF_\land$ and $SF_\lor$ are computed once and are not meant to change afterwards. In the main loop, the algorithm takes a sequent $(a, b)$ that has been proven but whose consequences have not yet been computed. Then, the map $P_{Cut}$ is updated twice. Its invariant is as follows:
    For every subformula $a$, $P_{Cut}(a)=\{b\mid (\getInverse(a),b)\in P\}$. Since $(a,b)\in P$, the invariant holds after each iteration. The situation is symmetric for $P_{Cut}(\getInverse(b))$.
    The map $P_\land$ is also updated according to the following invariant: for all subformulas $b, \psi$, $P_\land(b)(\psi)=\{a\land \psi\mid(a,b)\in P\text{ and }a\land \psi\in SF\}$. By construction, $\psi \in SF_\land(a)$ implies $a\land \psi\in SF$. Furthermore, since $(a,b)\in P$, the invariant holds. The reasoning is symmetric for $P_\land(b)$. Once the data structures are updated, we deduce new sequents. There are three cases (and their duals) to handle:
    \begin{enumerate}
        \item For sequents of the shape $(\phi,b)$, where $\phi$ belongs to the union of $P_{Cut}(a)$, $SF_\lor(a)$ and $P_\land(b)(a)$, we proceed set by set. If $\phi$ is in $SF_\lor(a)$, then we have $\phi=a\lor x$ for a given $x\in SF$ and $(a\lor x,b)$ is valid since it follows from the ($\vee$-$R$) rule. Otherwise, if $\phi\in P_{Cut}(a)$, then, according to $P_{Cut}$'s invariant, $(\getInverse(a),\phi)\in P$. Thus, also having $(a,b)$ valid, it is correct, using \textsc{Cut} rule, to deduce $(\phi,b)$. Finally, if $\phi=x\land a$ is in $P_\land(b)(a)$, then, according to $P_\land$'s invariant, $(x,b)\in P$ and $\phi\in SF$. Therefore, since $(a,b)$ and $(x,b)$ are valid, it is correct, using ($\wedge$-$R$) rule, to deduce $(x\land a,b)$.
        \item For sequents of the shape $(a,\psi)$ the reasoning is symmetric.
        \item In the last case, $a=b$, so by the \textsc{Replace} rule every sequent of the form $(a, \phi)$ is valid.
    \end{enumerate}
    Since the axiom and hypothesis sequents are valid sequents, the initialization of the worklist and $P$ is correct as well and hence, by induction, all sequents in $P$ are provable.
    \\
    ($\Leftarrow$) Conversely, suppose that the sequent $s$ has a proof from the axioms $A$. We proceed by induction on the proof. To show that provable sequents are eventually added to $P$, if the proof is an instance of \textsc{Ax} or of \textsc{Hyp}, $s$ is in $P$ by initialization, thus the algorithm yields \lstinline|true|. Then:
    \begin{enumerate}
        \item (\textsc{Cut} rule) Suppose $s=(a,b)$ and the premises $s_1=(a,c)$ and $s_2=(\getInverse(c),b)$. By induction, at some point $s_1,s_2\in P$. If $s_1$ has been added to $P$ before $s_2$, we have that $a\in P_{Cut}(\getInverse(c))$. Then when adding $s_2$ to $P$, we deduce, for all $\phi\in P_{Cut}(\getInverse(c))$, the sequent $(\phi,b)$. Hence $s$ gets added to the worklist. Hence, $s \in P$. The reasoning supposing $s_2$ is added to $P$ before $s_1$ is symmetric.
        \item (($\wedge$-$R$)) Suppose $s=(a\land k,b)$ and the premises $s_1=(a,b)$ and $s_2=(k,b)$. By induction, at some point $s_1,s_2\in P$. If $s_1$ has been added to $P$ before $s_2$, we have that $a\land k\in P_\land(b)(k)$, since $SF_\land(a)\subseteq SF$ and $s$ is the goal sequent. Then, the treatment of \textsf{add}$(s_2)$ deduces, for all $\phi\in P_\land(b)(k)$, the sequent $(\phi,b)$. Since $a\land k\in P_\land(b)(k)$, $s$ gets added to the worklist (and later in $P$)
        \item (($\vee$-$R$)) Suppose $s=(a,b\lor k)$ and the premise is $s_1=(a,b)$. By induction, at some point $s_1\in P$ and then the algorithm has deduced, for all $\psi\in SF_\lor(b)$, the sequent $(a,\psi)$. Since $k\in SF_\lor(b)$, $s$ is added to the worklist. The reasoning when $s_1=(a,k)$ is symmetric.
        \item (\textsc{Replace} rule) Suppose $s=(a,b)$ and the hypothesis $s_1=(a,a)$. By induction hypothesis, $s_1\in P$. Thus, the algorithm has deduced, for all $\phi\in SF$, the sequent $(a,\phi)$, and in particular $s$.
    \end{enumerate}
\end{proof}

\begin{theorem}[Complexity of entailment algorithm]
    The algorithm in \autoref{alg:olEntailment} has a time complexity in $\mathcal{O}(n^2(1+|A|))$ and space complexity $\mathcal{O}(n^2)$.
\end{theorem}
\begin{proof}
    We divide the complexity analysis into two parts: the initialization, and the main loop. Clearly, the initialization of every data structure takes time either constant or $\mathcal{O}(|SF|)$, i.e., linear in the size of the input formulas.

    For the main loop, first observe that its body will be executed at most once for every sequent $s=(a,b)$ where $a, b \in SF$. We analyze the time taken by the body of the loop across all sequents $s=(a, b)$, and split it into two parts: updating the data structure and adding new sequents to the worklist.

    The $P_{Cut}$ updates take average constant time for every $(a, b)$.
    Then, the $\SF_\land$ loop iterates over up to $|\SF|$ elements, but observe that (since each conjunction contributes to exactly two sets)
    $$\sum_{b \in \SF} |\SF_\land(b)| = \mathcal{O}(|\SF|)$$
    Hence, for a fixed $a$, the cost of this loop across all calls \lstinline|add(a, b)| is $\mathcal{O}(|\SF|)$ and hence the total cost across all possible sequents is $\mathcal{O}(|\SF|^2)$. The other $\SF_\land$ loop is similar.

    Then, we deduce new sequents. The key observation here is that for every particular way of deducing a sequent $s$ (that is, particular deduction rule and premises), there is exactly one time we will check if $s$ has already been proven, and add it to the worklist if not. Indeed, let us count for every sequent $s$ the number of times that the condition \lstinline|if !proven.contains(s)| is checked. Note that here we are not summing over all executions of the body with unique $(a, b)$, but combining all executions and summing over the $s$ that we possibly try to deduce, that is:
    \begin{enumerate}
        \item (\textsc{Cut}) A sequent $(\phi, b)$ is deduced 1 time for every $a$ such that $(a, b)$ is deduced and $\phi \in P_{Cut}(a)$, so at most $|A|$ times since  $P_{Cut}(a)$ is non-empty only when $a$ is an axiom formula.
        \item (($\wedge$-$R$)) A sequent $(\phi_1 \land \phi_2, b)$ is deduced 1 time only from either $(\phi_1, b)$ or $(\phi_2, b)$ (the second that is proven), since $\phi \in P_\land(b)(a)$ only if $a=\phi_1$ or $a=\phi_2$.
        \item (($\vee$-$R$)) A sequent $(\phi_1 \lor \phi_2, b)$ is deduced 1 time at most from each of $(\phi_1, b)$ or $(\phi_2, b)$  since $\phi \in \SF_\lor(a)$ only if $a=\phi_1$ or $a=\phi_2$.

        \item Again at most the same number of symmetric ways.

        \item (\textsc{Replace}) A sequent $(\phi, \psi)$ is deduced 1 time from each of $(\phi, \phi)$ and $(\psi, \psi)$, so 2 times in total
    \end{enumerate}
    Hence every sequent $(\phi, \psi)$ can be deduced in $\mathcal{O}(1+|A|)$ ways, hence the total number of deduction attempts across all executions of the body is $\mathcal{O}(|\SF|^2(1+|A|))$, and so the total cost across all executions of the body is $\mathcal{O}(|\SF|^2(1+|A|))$.

    Note that if we had simply computed the worst-case complexity over a single execution of the loop and multiplied by $n^2$, we would have obtained a worse bound. For example, for a sequent $(a, a)$ taken from the worklist, the execution of the loop takes linear time in $|\SF|$ to deduce all sequents of the form $(a, \phi)$, so we would only obtain a cubic bound.
    \vspace{0.5em}

    For space complexity, note that every data structure contains at most $\mathcal{O}(n^2)$ elements. For $P_\land$, note that, for any fixed $b$, if a formula $\phi \in SF$ is contained in both $P_\land(b)(k_1)$ and $P_\land(b)(k_2)$, then $\phi = k_1 \land k_2$. Hence, $\sum_{k\in SF} P_\land(b)(k) \leq 2|SF|$.
\end{proof}

\paragraph{Note on complexity of memoization and equality testing} The complexity proof of \autoref{alg:olEntailment} relies on the fact that operations on data structures such as hash maps and hash sets can be done in average constant time, but this is a simplification. In reality, adding or removing an element from a hash set or hash map requires checking equality with other elements to ensure that there are no hash collisions. However, checking whether two formulas given as syntax trees are equal takes linear time, and hence would add a linear factor to the complexity of the algorithm.

Luckily, this can be avoided by using reference equality instead of structural equality.
Recall that by the subformula property, the algorithm only ever sees the $\mathcal{O}(n)$ different subformulas of the input and their negation normal form. We can fix one particular object in memory for each of these formulas, and assign a different unique identifier to each. Such identifiers require only $\mathcal{O}(\log(n))$ bits to represent. Then, if two terms have the same identifier, they must be structurally equal. Additionally, we can (but do not need to) enforce that structurally equal formulas share the same object in memory by memoizing the construction of formulas, a strategy known as hash-consing.

Concretely, we must ensure that throughout the algorithm we do not create new formula objects.
In particular, computing \lstinline|getInverse| $=$ \lstinline|nnf($\neg$a)| creates a new node \lstinline|$\neg$a|, and then again $\mathcal{O}(|a|)$ new nodes when computing its negation normal form. To avoid this, we define \lstinline|getInverse| for formulas in negation normal form by recursion on the formula structure, memoizing the result in a field of each formula object.
This guarantees we never instantiate any node beyond the $n$ subnodes of the original formula (in negation normal form) and their inverse for a total of $2n$ nodes.

\section{A new way to generate hard synthetic SAT problems}
\label{sec:olBenchmarks}

We now study a particular class of satisfiability problems corresponding to verifying the equivalence between a Boolean circuit and its OL-normal form. A Boolean circuit is the representation of a formula $\phi$ as a directed acyclic graph (DAG) where each node is a logical operator (AND, OR, NOT) and the edges represent the flow of information, possibly enabling \textit{structure sharing}. Remember that we note $\NFOL(\phi)$ the OL-normal form of $\phi$. We then create, for the formula $\phi$ the problem $\phi' := \phi \leftrightarrow \NFOL(\phi) $. Note that this formula is always valid no matter what the original $\phi$ is. 

We apply this construction to a collection of circuits computing arithmetic functions from the EPFL Combinational Benchmark \cite{amaruEPFLCombinationalBenchmark2015}. Since the circuits themselves are describe computations, they are typically trivially satisfiable and invalid (otherwise they would represent constant functions true or false). Additionally, they each contain many output bits, so each circuits actually yields many different problems $\phi'$, often of increasing size. Our goal is to compare the runtime of our improved algorithm from \autoref{sec:improvedOLAlgo} in comparison with the state-of-the-art SAT solver Kissat~\cite{BiereFallerFazekasFleuryFroleyksPollitt-SAT-Competition-2024-solvers}.

We generate our benchmark the following way: For each circuit file (for example, \textsf{sqrt.aig}) and output bit $n$, we extract the formula $\phi_n$ corresponding to the ``cone'' of the $n$-th output bit of the circuit. We then compute the \OL-normal form $\NFOL(\phi_n)$ and export it to a second Aiger file. We also export a third file in DIMACS format containing clauses representing the Tseitin encoding of the formula $\phi_n \leftrightarrow \NFOL(\phi_n)$. We then measure the time taken by our implementation of orthologic proof search in C to import and prove the sequent $\phi \vdash \NFOL(\phi)$, and the time taken by Kissat to decide the satisfiability of the corresponding set of clauses. Results are reported in \autoref{fig:ol_kissat_arith} and \autoref{fig:ol_kissat_max}.

\begin{figure}[htbp]
  \centering
  \begin{subfigure}[t]{0.49\linewidth}
    \includegraphics[width=\linewidth]{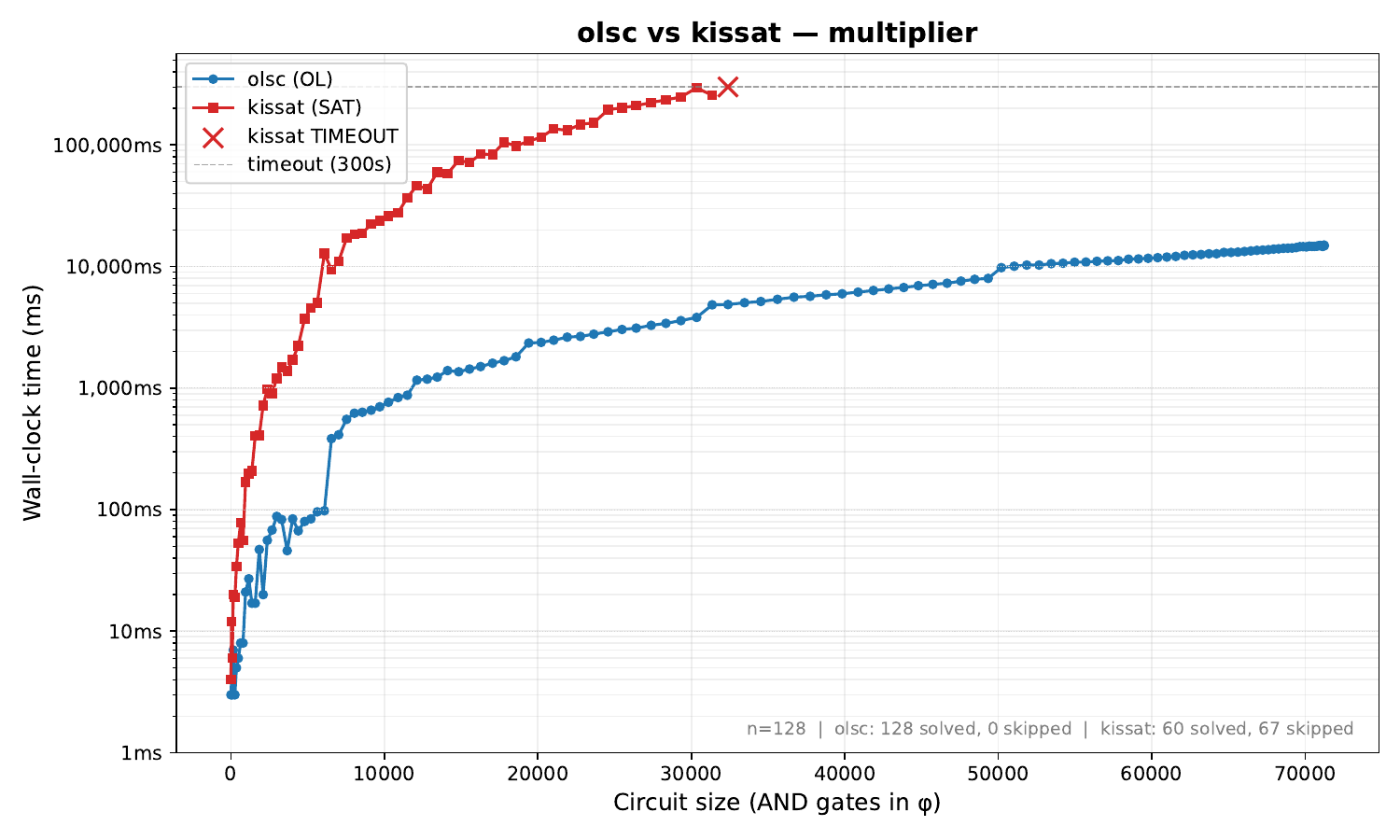}
    \caption{\textsf{multiplier}: \OL solves all 128 bits; Kissat times out on 67.}
    \label{fig:ol_kissat_multiplier}
  \end{subfigure}
  \hfill
  \begin{subfigure}[t]{0.49\linewidth}
    \includegraphics[width=\linewidth]{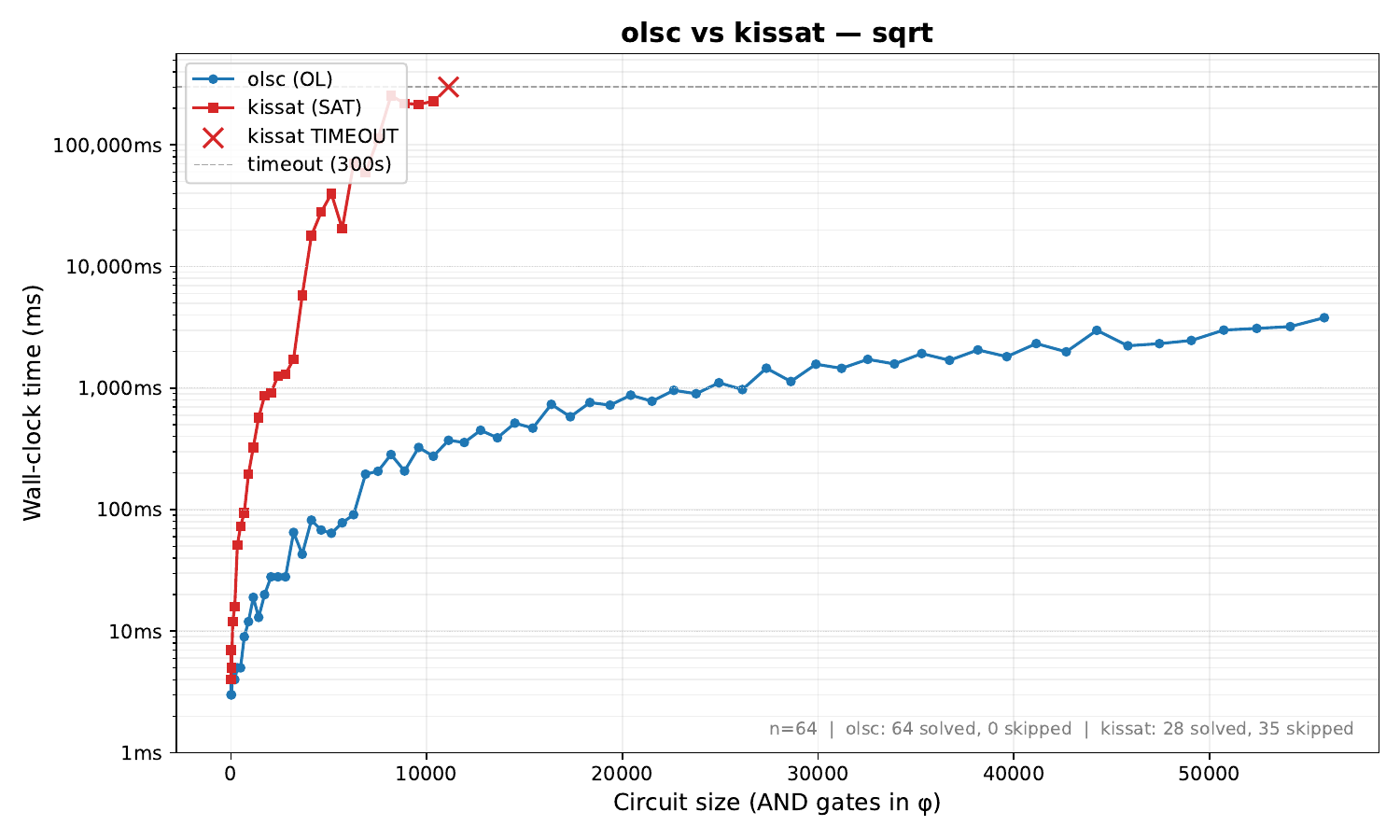}
    \caption{\textsf{sqrt}: \OL solves all 64 bits; Kissat times out on 35.}
    \label{fig:ol_kissat_sqrt}
  \end{subfigure}

  \medskip

  \begin{subfigure}[t]{0.49\linewidth}
    \includegraphics[width=\linewidth]{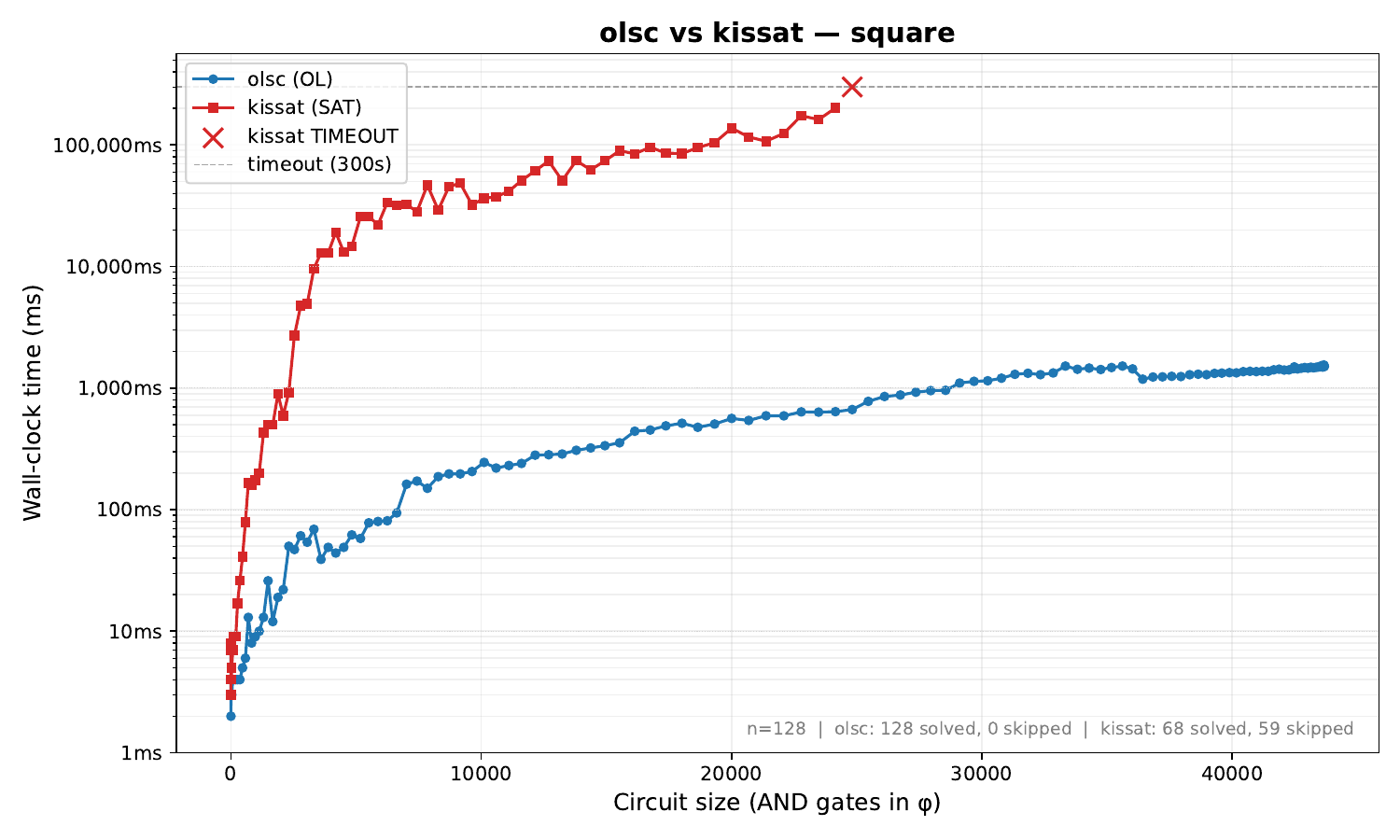}
    \caption{\textsf{square}: \OL solves all 128 bits; Kissat times out on 59.}
    \label{fig:ol_kissat_square}
  \end{subfigure}
  \hfill
  \begin{subfigure}[t]{0.49\linewidth}
    \includegraphics[width=\linewidth]{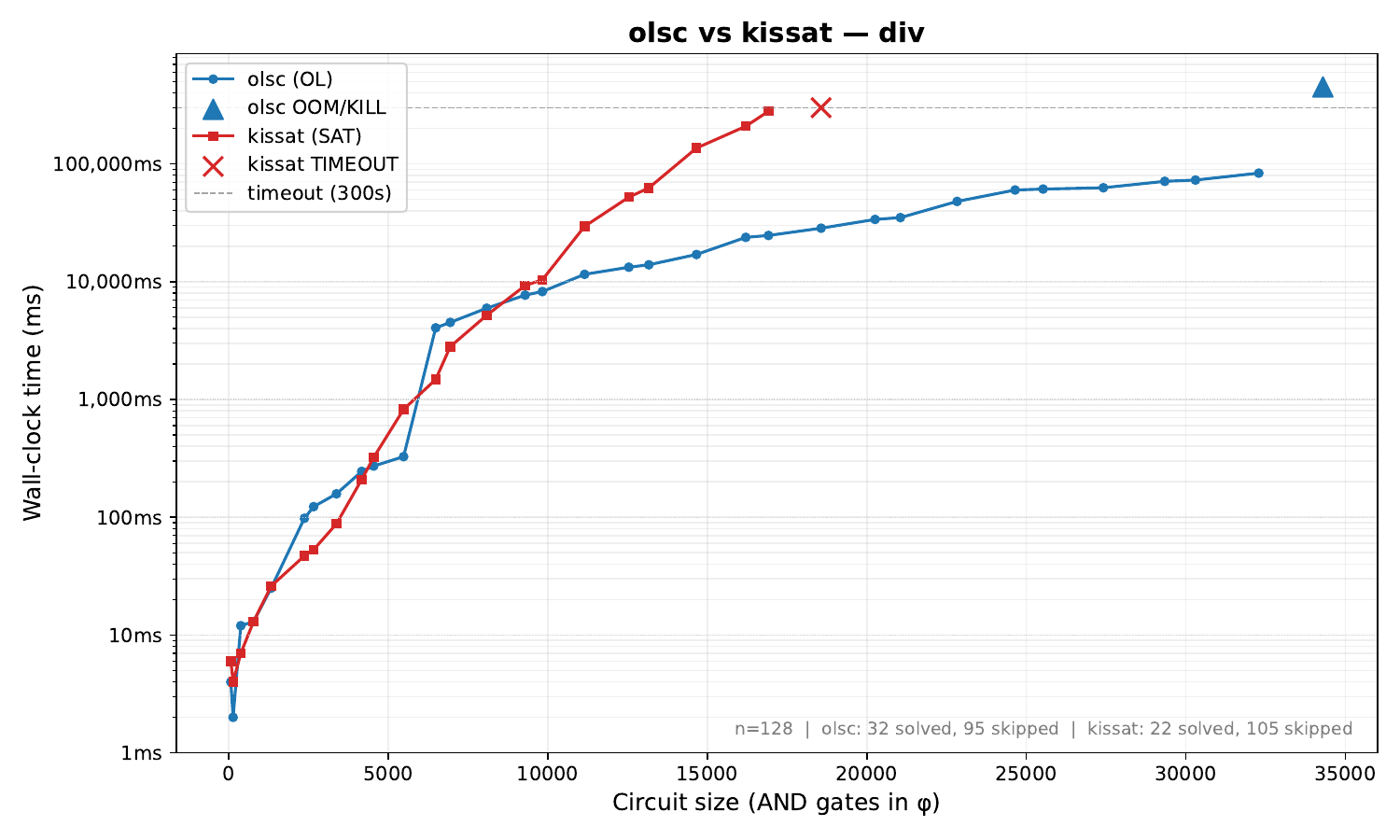}
    \caption{\textsf{div}: \OL solves 32 bits (OOM above $30\,\text{K}$ gates); Kissat solves 22 (timeout above $16\,\text{K}$ gates).}
    \label{fig:ol_kissat_div}
  \end{subfigure}
  \caption{Runtime comparison of \OL and Kissat on the $\phi \vdash \NFOL(\phi)$ benchmark for arithmetic circuits from the EPFL suite~\cite{amaruEPFLCombinationalBenchmark2015}. The $x$-axis shows circuit size (number of AND gates in $\phi$); the $y$-axis shows wall-clock time in milliseconds (logarithmic). A red cross marks a Kissat timeout (300\,s); a blue triangle marks an \OL out-of-memory (OOM) event.}
  \label{fig:ol_kissat_arith}
\end{figure}

\begin{figure}[htbp]
  \centering
  \includegraphics[width=0.65\linewidth]{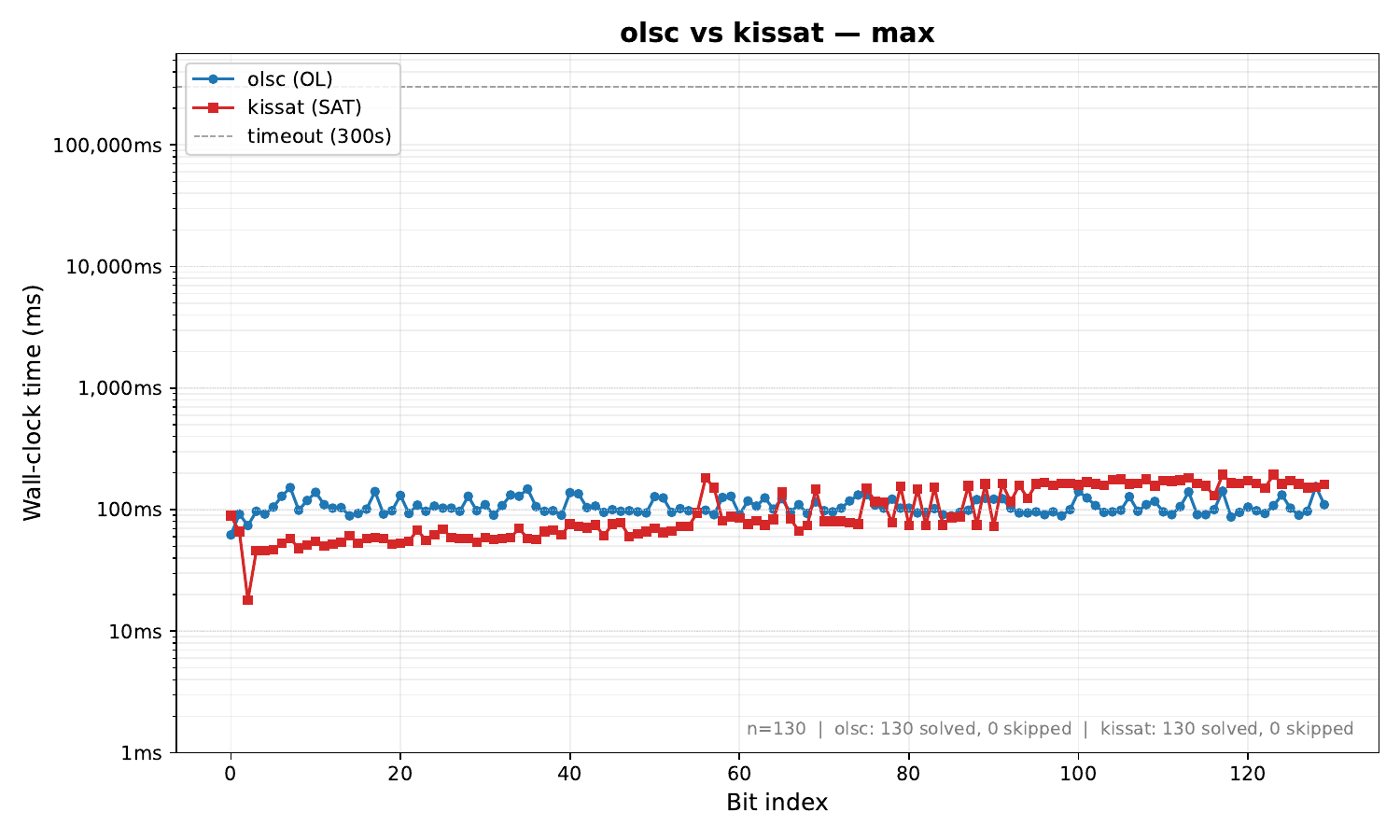}
  \caption{Runtime comparison on \textsf{max}. All 130 output bits have small cones (below $1\,000$ AND gates); both \OL and Kissat solve every instance in under 200\,ms with comparable runtimes. The $x$-axis here is the bit index rather than circuit size, as all cones are of similar size.}
  \label{fig:ol_kissat_max}
\end{figure} 

\paragraph{Easy and hard instances.}
On the three arithmetic circuits \textsf{multiplier}, \textsf{sqrt}, and \textsf{square}, \OL solves every output bit within the time budget, whereas Kissat times out on a large fraction of instances: 67 out of 128 bits for \textsf{multiplier}, 35 out of 64 for \textsf{sqrt}, and 59 out of 128 for \textsf{square} (\autoref{fig:ol_kissat_multiplier}, \autoref{fig:ol_kissat_sqrt}, \autoref{fig:ol_kissat_square}).
On the \textsf{div} circuit (\autoref{fig:ol_kissat_div}), instances are larger and both tools encounter their limits. Kissat times out beyond $\sim\!16\,000$ AND gates. \OL continues further but runs out of memory (OOM) on the very largest output cones (above $30\,000$ gates).
In contrast with the result presented here which uses the algorithm presented in \autoref{sec:proofsearch}, a C implementation of orthologic proof search based on preprocessing of the relations between sequents and data flow, similar to the Horn clause reduction described in \autoref{sec:preliminaries}, is orders of magnitude slower, and requires prohibitive amount of memory for all but the simplest circuits. This demonstrates the importance of the algorithmic improvements described in \autoref{sec:proofsearch} to be able to solve these benchmarks.

We intentionally leave out some of the circuits from the EPFL suite from our report. Three of them (hyp, log and sin) are consistently too big and neither tool can prove even the smallest output bit within the time budget. Three others (adder, bar and max) sit on the other side of the spectrum and are too small: both tools solve every output bit in under 200\,ms, with comparable runtimes (\autoref{fig:ol_kissat_max} shows the results for \textsf{max}).

For the four problems that end up in the sweet spot, the runtime curves show an asymptotic advantage of \OL. The algorithm cannot, by construction, deduce more than $\O(n^2)$ inequalities (corresponding to clauses). This means that there must exist a short solution to the unsatisfiability of the miter $\phi \wedge \neg\NFOL(\phi)$ but Kissat does typically not find it. Hence, constructing benchmarks of the form $\phi \leftrightarrow \NFOL(\phi)$ suggests a new attack angle to produce unsatisfiable benchmarks that are not well-covered by heuristics of SAT solvers.

\section{Preprocessing}
\label{sec:preprocessing}
In this final section, we describe a different orthologic-based approach to augment SAT solvers. The idea is to normalize the original formula into a normal form with respect to ortholattices laws, before solving them. As shown by researchers in recent work~\cite{DBLP:conf/cav/GuilloudBMK23}, orthologic formulas admit a normal form  that can be computed in worst-case quadratic time. Furthermore, this normal form is \textit{minimal}: it is always the smallest formula in the equivalence class of the input. Since the formula is smaller, and possibly has fewer variables, it is reasonable to expect that, in general, the resulting formula will be easier to solve for SAT solvers. Of course, there is no guarantee that this is the case; because SAT solvers use many heuristics, it is even possible that the transformation happens, by ill luck, to make the SAT solver slower. However, we will show  that, on average over a large set of benchmarks, it significantly improves solving time.
\paragraph{Methodology}
For every circuit in our benchmark, we compute its OL normalization. We then clausify the original circuit as well as the normalized one, and give them to a SAT solver,  Kissat~\cite{BiereFallerFazekasFleuryFroleyksPollitt-SAT-Competition-2024-solvers}. We justify that choice by the fact that Kissat is one of the top SAT solvers at the moment. For every circuit, we compute the size of the circuit before OL-normalization, the size after OL-normalization and the time taken by OL-normalization and by Kissat on both variants of the circuit.
\subsection{Evaluation}
To evaluate the impact of the OL-normalization algorithm as a preprocessing algorithm, we gathered different kinds of logical circuits that are used for real purposes. Notice that those circuits we use are given in a non-clausal form. The first kind is a set of 22 formulas, generated in 2001 for the formal verification of correct super-scalar microprocessors~\cite{velev}, 21 of them are unsatisfiable, and the remaining one is satisfiable. In the following table, only 7 of them appear in Table~\ref{tab:average}, since the OL-normalization timed out on the others. The second kind is a set of several hundred of bit vectors problems published in 2009 among the AIGER benchmarks~\cite{BiereAigerBenchmarks}, of which we select the 400 largest circuits. Finally, the last kind is a set of cryptographic boolean circuits that model known plain-text attacks on bounded-round versions of DES~\cite{tjunttil}. However, the SAT solver timed out on most of those instances in both cases; their results were not informative, so we omit them from the results. We also filter out too simple problems where Kissat is already very fast. All together, this gives us 410 problems on which we applied the method previously explained to measure the impact of the OL-normalization algorithm. Table~\ref{tab:average} summarizes the resulting data and shows the average on problems for which there has not been any timeout. Detailed results are available in Tables \ref{tab:preprocessingBenchmarks0} to \ref{tab:preprocessingBenchmarks3} in \autoref{sec:appendix}.
\begin{table}[htbp]
    \begin{center}\resizebox{\textwidth}{!}{%
        \begin{tabular}{| c H | c | c | c | c | c | c | c |}
        \hline
            &size&OL size&OL time&Kissat&Kissat/OL&Kissat/OL+norm&speed-up&speed-up with norm\\
            \hline
            average&9726&6322&4.66&59.36&50.76&55.43&0.16&0.07\\\hline
            \end{tabular}}
        \caption{Average circuit size and time over all circuits (see Appendix for more results).}
        \label{tab:average}
    \end{center}
\end{table}
The column \textsf{size} shows the size (in nodes) of the formula. The columns \textsf{Kissat} and \textsf{Kissat/OL} show the time for Kissat to decide respectively the formula without OL-normalization and the formula OL-normalized. The column \textsf{Kissat/OL+norm} shows the total time of OL-normalization and the solving of the formula OL-normalized. In many cases, the simplified circuit has actually more gates than the input circuit; this is due to the fact that reusing sub formulas that already exist somewhere else in the circuit can be more efficient than replacing them by simplified but new subformulas.
The last two columns \textsf{speed-up} and \textsf{speed-up with norm} show the speed-up respectively without the OL-normalization time and with it. We remark a net average speed-up of on average 7\% when counting the normalization time, and 16\% purely on SAT solving time.

\section{Conclusion}
\label{sec:conclusion}
We introduced a new proof-search algorithm for orthologic with axioms that is in practice significantly faster than previously existing algorithms, while retaining the same worst-case $\mathcal{O}(n^2(1+|A|))$ complexity.

We also identified and studied a new family of synthetic SAT benchmarks arising from the construction $\phi \leftrightarrow \NFOL(\phi)$: for any formula $\phi$, this equivalence is a tautology, yet its Tseitin encoding yields unsatisfiable instances that can be hard for state-of-the-art SAT solvers. We applied this construction to output bits of EPFL arithmetic circuits and showed that our new algorithm solves them efficiently while Kissat times out on a significant fraction. Since orthologic proofs are at most quadratic in size, short proofs exist for all these instances — yet SAT solver heuristics fail to find them. This construction hence provides a principled way to generate unsatisfiable benchmarks that expose a gap in CDCL solving.

Moreover, we evaluated the usefulness of the quadratic-time orthologic normalization as a preprocessing technique on satisfiability benchmarks in circuit form. We demonstrated that, on average, problems are solved approximately 7\% faster after OL-normalization on our benchmark set. These results suggest that orthologic entailment and normalization are promising building blocks for SAT solving pipelines, and we believe deeper integration with modern SAT solvers is a fruitful direction for future work.

\bibliography{sguilloud.bib, biblio.bib, vkuncak.bib}

\appendix
\section{Appendix: Detailed Benchmark Results}
\label{sec:appendix}
\begin{table}[htbp]
    \begin{center}\resizebox{\textwidth}{!}{%
        \begin{tabular}{| c | c H | c | c | c | c | c | c |}
        \hline
            Problem&Size&OL size&OL time&Kissat&Kissat/OL&Kissat/OL+norm&Speed-up&Speed-up with norm\\
            \hline
            4pipe&5274&6490&5,00&2,4686&1,6114&6,6114&0,5319&-0,6266\\\hline
            5pipe&9264&8223&16,00&3,9337&1,7647&17,7647&1,2292&-0,7786\\\hline
            6pipe&16170&18910&41,00&27,2784&12,3026&53,3026&1,2173&-0,4882\\\hline
            cvs-vc81759&3662&6182&0,90&17,8101&17,2549&18,1582&0,0322&-0,0192\\\hline
            cvs-vc81760&3662&6182&0,97&17,6892&17,1864&18,1539&0,0293&-0,0256\\\hline
            des-4-1-10&2556&4828&7,00&64,7531&83,9051&90,9086&-0,2283&-0,2877\\\hline
            des-4-1-11&2553&4825&6,75&23,9078&46,6821&53,4344&-0,4879&-0,5526\\\hline
            des-4-1-12&2556&4828&7,02&29,5854&41,5785&48,6027&-0,2884&-0,3913\\\hline
            des-4-1-13&2552&4824&7,56&31,2699&60,5110&68,0694&-0,4832&-0,5406\\\hline

            \end{tabular}}
        \caption{Measures on problems before and after OL normalization}
        \label{tab:preprocessingBenchmarks0}
    \end{center}
\end{table}

\begin{table}[htbp]
    \begin{center}\resizebox{\textwidth}{!}{
        \begin{tabular}{| c | c H | c | c | c | c | c | c |}
            \hline
            Problem&Size&OL size&OL time&Kissat&Kissat/OL&Kissat/OL+norm&Speed-up&Speed-up with norm\\
            \hline
            des-4-1-14&2546&4818&7,56&58,8476&49,2691&56,8276&0,1944&0,0355\\\hline
            des-4-1-15&2556&4828&7,29&31,0774&81,6611&88,9508&-0,6194&-0,6506\\\hline
            des-4-1-16&2549&4821&7,41&9,1072&39,1284&46,5358&-0,7672&-0,8043\\\hline
            des-4-1-17&2555&4827&7,37&41,7880&69,5035&76,8719&-0,3988&-0,4564\\\hline
            des-4-1-18&2553&4825&7,38&35,1534&38,0848&45,4601&-0,0770&-0,2267\\\hline
            des-4-1-19&2554&4826&7,19&51,2722&83,1188&90,3110&-0,3831&-0,4323\\\hline
            des-4-1-1&2547&4819&7,72&53,2462&68,6359&76,3569&-0,2242&-0,3027\\\hline
            des-4-1-20&2559&4831&6,66&36,2918&55,8569&62,5218&-0,3503&-0,4195\\\hline
            des-4-1-21&2554&4826&7,05&69,2210&67,0272&74,0788&0,0327&-0,0656\\\hline
            des-4-1-2&2556&4828&6,97&25,1371&99,2928&106,2619&-0,7468&-0,7634\\\hline
            des-4-1-3&2557&4829&6,64&61,7551&77,1045&83,7470&-0,1991&-0,2626\\\hline
            des-4-1-4&2559&4831&6,72&25,1765&57,9580&64,6792&-0,5656&-0,6107\\\hline
            des-4-1-5&2554&4826&7,02&62,9419&90,5743&97,5941&-0,3051&-0,3551\\\hline
            des-4-1-6&2551&4823&6,93&83,5780&77,4677&84,3977&0,0789&-0,0097\\\hline
            des-4-1-7&2550&4822&7,60&11,1279&18,6076&26,2058&-0,4020&-0,5754\\\hline
            des-4-1-8&2552&4824&7,58&63,5239&44,8265&52,4088&0,4171&0,2121\\\hline
            des-4-1-9&2555&4827&7,16&24,8501&20,0497&27,2067&0,2394&-0,0866\\\hline
            des-4-2-10&4144&7788&7,13&27,2010&53,4352&60,5688&-0,4910&-0,5509\\\hline
            des-4-2-11&4121&7737&7,27&14,0401&19,1189&26,3853&-0,2656&-0,4679\\\hline
            des-4-2-12&4149&7793&7,43&23,2885&50,5641&57,9912&-0,5394&-0,5984\\\hline
            des-4-2-13&4152&7794&7,49&75,3467&15,6686&23,1604&3,8088&2,2533\\\hline
            des-4-2-14&4154&7800&7,83&19,9972&52,7515&60,5773&-0,6209&-0,6699\\\hline
            des-4-2-15&4161&7807&7,72&23,3266&66,0310&73,7542&-0,6467&-0,6837\\\hline
            des-4-2-16&4118&7736&7,39&31,8417&31,0057&38,3991&0,0270&-0,1708\\\hline
            des-4-2-17&4155&7801&7,53&43,5668&48,2375&55,7707&-0,0968&-0,2188\\\hline
            des-4-2-18&4147&7793&8,02&68,9178&45,7864&53,8027&0,5052&0,2809\\\hline
            des-4-2-19&4147&7793&7,80&47,6810&86,0082&93,8056&-0,4456&-0,4917\\\hline
            des-4-2-1&4117&7735&7,15&56,1578&33,9999&41,1507&0,6517&0,3647\\\hline
            des-4-2-20&4132&7752&7,22&67,2728&61,6027&68,8181&0,0920&-0,0225\\\hline
            des-4-2-21&4152&7796&7,39&80,0591&32,4701&39,8553&1,4656&1,0087\\\hline
            des-4-2-2&4151&7791&7,04&84,7390&19,9786&27,0144&3,2415&2,1368\\\hline
            des-4-2-3&4158&7806&7,24&21,0431&20,5398&27,7755&0,0245&-0,2424\\\hline
            des-4-2-4&4160&7808&7,78&39,6752&41,6503&49,4320&-0,0474&-0,1974\\\hline
            des-4-2-5&4160&7804&7,74&39,6361&17,1189&24,8620&1,3153&0,5942\\\hline
            des-4-2-6&4148&7792&7,25&6,3511&30,6947&37,9444&-0,7931&-0,8326\\\hline
            des-4-2-7&4148&7792&7,60&37,6652&76,8233&84,4218&-0,5097&-0,5538\\\hline
            des-4-2-8&4152&7796&7,20&44,6267&51,0165&58,2211&-0,1252&-0,2335\\\hline
            \end{tabular}}
        \caption{Measures on problems before and after OL normalization}
        \label{tab:preprocessingBenchmarks1}
    \end{center}
\end{table}
\begin{table}[htbp]
    \begin{center}\resizebox{\textwidth}{!}{
        \begin{tabular}{| c | c H | c | c | c | c | c | c |}
            \hline
            Problem&Size&OL size&OL time&Kissat&Kissat/OL&Kissat/OL+norm&Speed-up&Speed-up with norm\\
            \hline
            des-4-2-9&4149&7795&7,24&83,3318&80,6225&87,8636&0,0336&-0,0516\\\hline
            des-5-1-17&3309&6333&0,98&1000,0055&412,0046&412,9809&1,4272&1,4214\\\hline
            des-5-1-20&3311&6335&1,07&1000,0072&353,3140&354,3886&1,8304&1,8218\\\hline
            fig6.phx&9370&17775&3,48&0,1587&0,1482&3,6328&0,0705&-0,9563\\\hline
            icbrt&10505&19822&22,36&16,1837&0,0126&22,3761&1286,2816&-0,2767\\\hline
            icbrtinvalidvc&11053&20486&23,49&0,0139&0,0121&23,4988&0,1417&-0,9994\\\hline
            icbrtor&10505&19822&22,35&16,1764&0,0108&22,3563&1496,2579&-0,2764\\\hline
            icbrtorinvalidvc&11053&20486&23,58&0,0143&0,0128&23,5933&0,1205&-0,9994\\\hline
            isqrtadd&5551&10638&8,56&143,8542&185,7079&194,2715&-0,2254&-0,2595\\\hline
            isqrtaddeqcheck&4095&8021&3,45&303,6890&275,8074&279,2547&0,1011&0,0875\\\hline
            isqrtaddnoif&5978&11572&8,24&154,4876&165,4805&173,7160&-0,0664&-0,1107\\\hline
            isqrtaddnoifinvalidvc&6237&11705&8,40&0,0101&0,0074&8,4112&0,3601&-0,9988\\\hline
            isqrt&5551&10638&8,46&144,5731&186,6313&195,0863&-0,2254&-0,2589\\\hline
            isqrteqcheck&4095&8021&3,41&304,8153&277,1961&280,6068&0,0996&0,0863\\\hline
            isqrtinvalidvc&5807&10749&8,77&0,0097&0,0069&8,7772&0,3926&-0,9989\\\hline
            isqrtnoif&5978&11572&8,18&154,5150&165,5195&173,7013&-0,0665&-0,1105\\\hline
            isqrtnoifinvalidvc&6237&11705&8,61&0,0097&0,0071&8,6206&0,3703&-0,9989\\\hline
            \end{tabular}}
        \caption{Measures on problems before and after OL normalization}
        \label{tab:preprocessingBenchmarks2}
    \end{center}
\end{table}
\begin{table}[htbp]
    \begin{center}\resizebox{\textwidth}{!}{
        \begin{tabular}{| c | c H | c | c | c | c | c | c |}
            \hline
            Problem&Size&OL size&OL time&Kissat&Kissat/OL&Kissat/OL+norm&Speed-up&Speed-up with norm\\
            \hline
            lfsr-002-015-016&1779&1873&0,46&0,0307&0,0328&0,4963&-0,0658&-0,9382\\\hline
            lfsr-002-015-032&3459&3617&1,03&0,0840&0,0751&1,1066&0,1188&-0,9241\\\hline
            lfsr-002-031-016&3411&3537&0,88&0,0534&0,0383&0,9159&0,3946&-0,9416\\\hline
            lfsr-004-015-016&3786&4100&1,09&0,1270&0,1381&1,2232&-0,0804&-0,8962\\\hline
            maxand016&1042&1430&0,34&0,1318&0,1070&0,4430&0,2311&-0,7025\\\hline
            maxand032&3138&4950&1,38&0,3046&0,3017&1,6843&0,0097&-0,8192\\\hline
            maxor008&854&1436&0,41&0,1388&0,1672&0,5796&-0,1698&-0,7606\\\hline
            maxor016&3150&5708&1,43&1,1789&1,4069&2,8358&-0,1621&-0,5843\\\hline
            maxor032&11966&22700&27,33&26,4611&25,0209&52,3534&0,0576&-0,4946\\\hline
            maxxor016&3214&5836&1,48&191,5726&183,8764&185,3568&0,0419&0,0335\\\hline
            maxxormaxorand016&5047&9160&2,40&154,9670&218,6255&221,0217&-0,2912&-0,2989\\\hline
            minand016&1058&1655&0,43&0,1474&0,2479&0,6819&-0,4056&-0,7839\\\hline
            minandmaxor008&956&1752&0,41&0,2125&0,1767&0,5870&0,2025&-0,6381\\\hline
            minandmaxor016&3604&6856&1,72&1,2004&1,6016&3,3188&-0,2505&-0,6383\\\hline
            minor032&3138&4950&1,32&0,3358&0,2951&1,6182&0,1378&-0,7925\\\hline
            minxor016&1108&1630&0,41&0,3967&0,4373&0,8467&-0,0927&-0,5314\\\hline
            minxor032&3268&5342&1,44&22,8997&10,1471&11,5842&1,2568&0,9768\\\hline
            minxorminand008&704&1272&0,33&0,3692&0,4313&0,7574&-0,1440&-0,5126\\\hline
            minxorminand016&2200&4120&0,76&5,5810&3,1965&3,9609&0,7460&0,4090\\\hline
            minxorminand032&7496&14424&7,28&249,3727&114,6178&121,8933&1,1757&1,0458\\\hline
            mulhs08&1092&2103&0,47&5,4162&5,6992&6,1715&-0,0497&-0,1224\\\hline
            mult-ub-8x8-1.sf&2457&4765&0,93&11,8968&11,5264&12,4573&0,0321&-0,0450\\\hline
            nextpoweroftwo032&552&672&0,42&0,0547&0,0507&0,4755&0,0779&-0,8850\\\hline
            nextpoweroftwo064&1128&1376&1,64&0,2916&0,1206&1,7644&1,4186&-0,8347\\\hline
            nlzbe016&1093&709&0,29&0,0467&0,0043&0,2966&9,7509&-0,8424\\\hline
            nlzbe032&4258&1838&0,56&0,1221&0,0042&0,5645&27,9365&-0,7837\\\hline
            nlzbsdown032&3336&1458&0,44&0,1003&0,0037&0,4481&26,1459&-0,7762\\\hline
            problem-15&1744&3344&0,60&27,9919&54,5389&55,1420&-0,4868&-0,4924\\\hline
            problem-1&3418&6752&3,29&0,0977&0,2340&3,5210&-0,5823&-0,9722\\\hline
            problem-22&8165&15862&1,75&69,7604&112,4134&114,1596&-0,3794&-0,3889\\\hline
            problem-2&3277&6499&5,25&0,0625&0,0051&5,2582&11,3281&-0,9881\\\hline
            problem-6&4350&7366&1,23&0,1514&0,2053&1,4341&-0,2625&-0,8944\\\hline
            problem-9&2461&4869&1,09&28,2475&2,2805&3,3701&11,3864&7,3818\\\hline
            smulov1bw12&1585&3133&0,60&382,2924&294,6817&295,2769&0,2973&0,2947\\\hline
            smulov1bw16&2847&5649&1,07&824,1497&859,3104&860,3844&-0,0409&-0,0421\\\hline
            smulov2bw016&2073&3960&0,68&339,9424&180,8193&181,5008&0,8800&0,8730\\\hline
            smulov4bw0032&1797&3485&0,48&0,0053&0,0033&0,4828&0,6314&-0,9890\\\hline
            smulov4bw0064&7685&15149&3,32&0,0097&0,0075&3,3311&0,3050&-0,9971\\\hline
            \end{tabular}}
        \caption{Measures on problems before and after OL normalization}
        \label{tab:preprocessingBenchmarks3}
    \end{center}
\end{table}

\end{document}